\documentclass[aps,prl,twocolumn,preprintnumbers,superscriptaddress,amsmath,amssymb,footinbib]{revtex4-2}
\usepackage{graphicx,epsfig}
\usepackage{bm}
\usepackage{dcolumn}
\usepackage[english]{babel}
\usepackage[utf8]{inputenc}
\usepackage{comment}
\usepackage{scalerel}
\usepackage{lipsum}
\usepackage{footnote}
\usepackage{amsmath}
\usepackage{amsfonts}
\usepackage{amssymb}
\usepackage{color}
\usepackage[usenames,dvipsnames]{xcolor}
\usepackage[colorlinks,bookmarks=false,citecolor=NavyBlue,linkcolor=Red,urlcolor=blue,%backref
]{hyperref}
\usepackage[export]{adjustbox}
\usepackage{amsthm}
\usepackage{simplewick}
\usepackage{bbold}
\usepackage{MnSymbol}
\usepackage{ wasysym }
\usepackage[title]{appendix}
\usepackage{mathrsfs}  
\usepackage{braket}
\usepackage{graphics}
\usepackage{todonotes}

\usepackage{subfigure}

\usepackage{multirow}
\usepackage{algorithm}
\usepackage{algpseudocode}

\newcommand{\be}{\begin{equation}}
\newcommand{\ee}{\end{equation}}
\newcommand{\bea}{\begin{eqnarray}}
\newcommand{\eea}{\end{eqnarray}}

\newcommand{\Id}{\mathbb{1}}
\newcommand{\nsamp}{\mathcal{N}}  % NUMBER OF SAMPLES 
\newcommand{\Trace}{\text{Tr}}

\newcommand{\StabGr}{G_{\mathcal{S}}}
\newcommand{\StabDim}{{k}}
\newcommand{\StabNull}{{\nu}}

\newcommand{\Hgate}{{\text{H}}}
\newcommand{\Sgate}{{\text{S}}}
\newcommand{\CNOT}{{\text{CNOT}}}
\newcommand{\Tgate}{{\text{T}}}
\newcommand{\Mod}[1]{\ (\mathrm{mod}\ #1)}

\newtheorem{lemma}{Lemma}

\date{ \today} 
\begin{document}

\title{Learning the stabilizer group of a Matrix Product State}
\author{Guglielmo Lami}
\affiliation{International School for Advanced Studies (SISSA), 34136 Trieste, Italy}
\affiliation{Laboratoire de Physique Théorique et Modélisation, CY Cergy Paris Université, CNRS, F-95302 Cergy-Pontoise, France}
\author{Mario Collura}
\affiliation{International School for Advanced Studies (SISSA), 34136 Trieste, Italy}
\affiliation{INFN Sezione di Trieste, 34136 Trieste, Italy}

\begin{abstract}
We present a novel classical algorithm designed to learn the stabilizer group -- namely the group of Pauli strings for which a state is a $\pm 1$ eigenvector -- of a given Matrix Product State (MPS). The algorithm is based on a clever and theoretically grounded biased sampling in the Pauli (or Bell) basis. Its output is a set of independent stabilizer generators whose total number is directly associated with the stabilizer nullity, notably a well-established nonstabilizer monotone.
We benchmark our method on $T$-doped states randomly scrambled via Clifford unitary dynamics, demonstrating very accurate estimates up to highly-entangled MPS with bond dimension $\chi\sim 10^3$. Our method, thanks to a very favourable scaling $\mathcal{O}(\chi^3)$, represents the first effective approach to obtain a genuine magic monotone for MPS, enabling systematic investigations of quantum many-body physics out-of-equilibrium.
\end{abstract}

\maketitle

\paragraph{Introduction. --} 
Quantum states of many interacting particles (or qubits) 
have in general a very high degree of complexity, due to exponential vastness of the Hilbert space~\cite{Feynman1982,RevModPhys.71.1253}. Achieving a full comprehension of which states can be simulated with reasonable classical computational resources, i.e.\ polynomial in the number $N$ of particles, is a task of invaluable importance since, for instance, it can shed light on areas where a quantum computational advantage might be realized~\cite{Feynman1982,shor,kitaev2002classical,Nielsen_chuang_2010}. There are at least two known general classes of states falling into this scenario.

First, states with sufficiently low amount of quantum correlations between the constituencies, i.e.\ low entanglement, can be simulated by means of Tensor Networks (TN) ~\cite{Silvi_2019,Schollwock_2011,biamonte2020lectures}. TN, among which Matrix Product States (MPS) play a primary role, represent the quantum wave function through a network of tensors whose indices correspond to both physical variables (such as spin), and auxiliary fictitious variables.
The latter are always contracted (summed over) and their role is to encode entanglement. In 1D, MPS can be used to simulate any state whose entanglement is bounded by a constant which does not scale extensively with the size $N$~\cite{Vidal_2004,Hastings_2007,Cirac_2021}.

Besides, there exists another general framework for simulating certain quantum states, which is linked to the specific structure of the Clifford group~\cite{Nielsen_chuang_2010,Gottesman_1997,Gottesman_1998_1,Gottesman_1998_2,Aaronson_2004,Dehaene_2003}. This group consists of unitary transformations mapping the group of strings of Pauli matrices over $N$ qubits (Pauli group) to itself under conjugation. When considering any state $\ket{\psi}$, it is possible to identify an abelian subgroup of the Pauli group, referred to as stabilizer group and denoted as $\StabGr(\ket{\psi})$. This is generated by $\StabDim_{\psi}$ mutually commuting Pauli strings $\pmb{\sigma}$ which stabilize the state, i.e.\ such that $\pmb{\sigma}\ket{\psi} = \pm \ket{\psi}$. States for which $\StabDim_{\psi}=N$ are dubbed stabilizers, and are uniquely identified by the list of Pauli generators~\cite{Nielsen_chuang_2010}. Stabilizer can be equivalently characterized as those states obtained by applying Clifford transformations to the computational basis state $\ket{0}^{\otimes N}$, and accordingly can encode arbitrary amounts of entanglement. Nevertheless, they can be simulated classically using the stabilizer formalism. Indeed, since $2$ classical bits can encode the $4$ Pauli matrices, one can conveniently store the generators of $\StabGr(\ket{\psi})$ in a tableau of size $N \times 2N$. Operations such as evaluating expectation values of Pauli operators, applying Clifford unitaries or measurements of Pauli operators, can be performed efficiently by updating the tableau at cost $\mathcal{O}(N^2)$~\cite{Aaronson_2004}. Nevertheless, to cover the entire Hilbert space (i.e.\ to achieve quantum computational universality), non-Clifford unitary transformations are needed and typically the complexity of the tableau algorithm scales exponentially with the number of these. The amount of non-Clifford resources needed to prepare a state is usually dubbed as nonstabilizerness and it is known as one of the veritable resources of the quantum realm~\cite{Winter_2022, Howard_2017,Leone_2022}.

Interestingly, until now only a few works have explored the interconnections between these two approaches. Recently, Stabilizer Renyi Entropies (SRE) gained significant interest as potential candidates for quantifying nonstabilizerness~\cite{Oliviero_2022_1, Leone_2022, Tirrito_2023, rattacaso2023stabilizer}, with the advantage of being amenable to experimental measurements ~\cite{Oliviero_2022_2, Gullans_2023}. Ways of evaluating SRE for MPS have been discussed in Ref.~\cite{Lami_2023, Haug_2023_1, Haug_2023_2, Tarabunga_2023_1}. However, as shown in Ref.~\cite{Haug_2023_2}, SRE are not genuine magic monotones (for qubits systems~\cite{Tarabunga_2023_2}) and thus the question of how to extract a veritable measure of magic from an MPS remains open. In this letter, we consider the task of learning the stabilizer group $\StabGr(\ket{\psi})$ of a state $\ket{\psi}$ given as MPS. The goal is achieved by employing a new type of sampling in the Pauli basis that is intentionally biased to favor the extraction strings belonging to the stabilizer group. To ensure a comprehensive mapping of the entire group $\StabGr(\ket{\psi})$, the sampling can be iterated over modified states obtained with Clifford transformations from the original MPS,
the computational cost for a single iteration being $\mathcal{O}(N\chi^3)$, where $\chi$ is the bond of the MPS. The final output of our Algorithm is a set stabilizer generators and an estimation of the stabilizer dimension of the state. Benchmarks show that this estimate is consistently accurate, i.e.\ one has a high probability of correctly learning all the generators of $\StabGr(\ket{\psi})$.

The task we consider is of crucial importance, in fact our novel Algorithm is the first known method to obtain a genuine magic monotone for MPS with reasonable computational resources. Consequently, it can pave the way to a systematic numerical investigation of the nonstabilizerness of quantum many-body states. In addition, one could in principle exploit the knowledge of the stabilizer group of the MPS to reduce its computational complexity. This could result into hybrid MPS stabilizer techniques, offering novel powerful means of classically simulating quantum systems. 
For instance, any state $\ket{\psi}$ with stabilizer dimension $\StabDim_{\psi}$ can be partially disentangled with a suitable Clifford unitary $U_{\mathcal{C}}$, obtaining $U_{\mathcal{C}} \ket{\psi} = \ket{0}^{\otimes \StabDim_{\psi}} \ket{\tilde{\psi}}$.~\cite{Nielsen_chuang_2010}\\

\paragraph{Preliminaries. --}
Let us consider a quantum system consisting of $N$ qubits.
We identify the Pauli matrices by $\{\sigma^{\alpha}\}_{\alpha=0}^{3}$, 
with $\sigma^{0} = \mathbb{1}$, and with $\pmb{\sigma} = \prod_{j=1}^{N}\sigma_{j} \in \mathcal{P}_{N}$ a generic $N-$qubits Pauli strings where $\mathcal{P}_{N} = \{ \sigma^0, \sigma^1, \sigma^2, \sigma^3 \}^{\otimes N}$. 
For a pure normalised state $\rho = \ket{\psi} \langle \psi|$, 
we can define the stabilizer group $\StabGr(\ket{\psi})$ as the set of Pauli strings in $\mathcal{P}_{N}$ for which $\ket{\psi}$ is an eigenstate, i.e.\ $\StabGr(\ket{\psi}) = \{ \pmb{\sigma} \text{ s.t. } \pmb{\sigma} \ket{\psi} = \pm \ket{\psi} \}$. $\StabGr$ is an abelian subgroup of the $N$ qubits Pauli group~\cite{Nielsen_chuang_2010}. \\
The stabilizer dimension $\StabDim_{\psi}$ of $\ket{\psi}$ is the number of independent (commuting) Pauli strings generating $\StabGr$. Equivalently, one can define $\StabDim_{\psi} \equiv \log_2 |\StabGr(\ket{\psi})|$, where $|\cdot|$ represents the cardinality of the set. The stabilizer nullity is defined as $\StabNull_{\psi}=N-\StabDim_{\psi}$ and is a genuine magic monotone since it is non-increasing under any stabilizer operations, such as Clifford unitaries or measurements of Pauli operators~\cite{Beverland_2020}. \\
By definition, $\StabDim_{\psi}=N$ for stabilizer states, whereas $\StabDim_{\psi} \geq N - t$ for $t-$doped stabilizer state (see Lemma~\ref{lemma:stabdim} in Supp.Mat.)~\cite{Grewal_2023}. These are states obtained from the computational basis state $\ket{0}^{\otimes N}$ through the application of a circuit consisting of Clifford gates and at most $t$ single-qubit non-Clifford $\Tgate$ gates ($t=0$ for stabilizers). The $\Tgate$ gate is defined as $\Tgate = \text{diag}(1, e^{i \pi/4})$, which together with Clifford gates $\{ \Hgate, \Sgate, \CNOT \}$ form an universal set of gates. Due to what has been mentioned, the number $t$ of $\Tgate$-gates required to synthesize a state $\ket{\psi}$ is lower bounded by $\StabNull_{\psi}$~\cite{Jiang_2023}. \\
In Ref.~\cite{Montanaro_2017}, Montanaro introduced a learning procedure for stabilizer states which is based on Bell sampling, i.e.\ joint measurements on two copies of the state in the Bell basis. Such measurements correspond to sampling Pauli strings $\pmb{\sigma} \in \mathcal{P}_{N}$ with a probability defined as $\Pi_{\rho}(\pmb{\sigma}) = \frac{1}{2^N} \Trace[\rho \pmb{\sigma}]^{2}$, and we will also refer to it as Pauli sampling. Montanaro showed that for stabilizer states the generators of $\StabGr(\ket{\psi})$ can be learned with only $\mathcal{O}(N)$ samples from $\Pi_{\rho}$~\cite{Montanaro_2017}, resulting in an exponential speedup for stabilizer states compared to tomographic methods~\cite{Anshu_2023}. The question of whether and how Montanaro's method can be expanded to learn $t-$doped states is a significant subject of research~\cite{Leone_2023, Grewal_2023,Hangleiter_2023}. \\
In Ref.~\cite{Leone_2023} the authors employ in turn the Bell sampling and notice that the extraction of $\StabDim_{\psi} + N$ random Pauli operators from $\StabGr(\ket{\psi})$ is sufficient to determine a generator set with failure probability of at most $2^{-N}$. However, for a $t-$dopes state the probability of obtaining a string $\pmb{\sigma} \in \StabGr(\ket{\psi})$ is $p(\pmb{\sigma} \in \StabGr) = |\StabGr(\ket{\psi})| \cdot \frac{1}{2^N} = 2^{\StabDim_{\psi} - N} \geq 2^{-t}$, since $\Pi_{\rho}(\pmb{\sigma}) = 1/2^N$ for any stabilizer string. This result shows that, in general, the probability of successfully finding a stabilizer string decreases exponentially with $t$. Hence, the approach in Ref.~\cite{Leone_2023} is feasible only when $t = \mathcal{O}(\log N)$. \\

\paragraph{MPS stabilizer sampling. --}
We consider a pure state $\ket{\psi}$ represented in the MPS form~\cite{Schollwock_2011,Silvi_2019,biamonte2020lectures} $\ket{\psi} = \sum_{s_1,s_2,\dots,s_N} \mathbb{A}^{s_1}_{1}\mathbb{A}^{s_2}_{2}\cdots\mathbb{A}^{s_N}_{N} |s_1,s_2,\dots,s_N \rangle$, with $\mathbb{A}^{s_{j}}_{j}$ being $\chi_{j-1}\times\chi_{j}$ matrices, except at the left (right) boundary where $\mathbb{A}^{s_{1}}_{1}$ ($\mathbb{A}^{s_{N}}_{N}$) is a $1\times\chi_1$ ($\chi_{N-1}\times1$) row (column) vector. Here $|s_{j}\rangle \in \{|0\rangle, |1\rangle\}$ is the local computational basis. Without loss of generality, the state is assumed right-normalised, namely $\sum_{s_{j}} \mathbb{A}^{s_{j}}_{j} (\mathbb{A}^{s_{j}}_{j})^{\dag} = \Id$. \\
As shown for the first time in Ref.~\cite{Lami_2023}, given the MPS $\ket{\psi}$ one can efficiently achieve a perfect Pauli sampling from the probability distribution $\Pi_{\rho}(\pmb{\sigma})$ at computational cost $\mathcal{O}(N \chi^3)$, with $\chi = \max_i \chi_i$. Such a sampling is achieved by exploiting the decomposition
\begin{equation}\label{eq:chain_prob}
\Pi_{\rho}(\pmb{\sigma}) = \pi_{\rho}(\sigma_1) \pi_{\rho}(\sigma_2|\sigma_1) ... \, \pi_{\rho}(\sigma_N|\sigma_1\cdots\sigma_{N-1}) \, ,
\end{equation}
where $\pi_{\rho}(\sigma_i|\sigma_{1}\cdots \sigma_{i-1}) = \pi_{\rho}(\sigma_1\cdots \sigma_i) / \pi_{\rho}(\sigma_1\cdots \sigma_{i-1})$ is the probability of Pauli matrix $\sigma_{i}$ appearing at position $i$ given that $\sigma_1 ... \, \sigma_{i-1}$ have been extracted at positions $1 ... \, i-1$, regardless of the occurrences in the remaining part of the system (i.e.\ marginalizing over Pauli sub-strings on qubits $i+1 \dots N$). Perfect sampling operates going through qubits $i=1,2... \, N$ in a sweep and sampling each local Pauli matrix according to the conditional probabilities $\pi_{\rho}(\sigma_i|\sigma_{1}\cdots \sigma_{i-1})$
~\cite{Stoudenmire_2010,Ferris_2012,Lami_2023}.

Here, we aim to identify the stabilizer group $\StabGr(\ket{\psi})$ of the MPS. In principle, one could perform a perfect sampling of Pauli strings $\pmb{\sigma}$ with the hope of sampling stabilizer strings. However, has outlined before, the probability of sampling a stabilizer string decreases exponentially as $\StabDim_{\psi}$ decrease. Nevertheless, in an MPS simulation a perfect sampling is not the sole option, and one can alternatively attempt to bias the extraction of stabilizer strings. 

With this spirit, we introduce a novel sampling strategy. In this approach, at a generic step $i$ of the sweep, a certain number $K$ of sub-strings $\{ \pmb{\sigma}_{[1,i]}^{\mu} \}_{\mu=1}^K$ are stored (here and in the following we use $\pmb{\sigma}_{[1,i]}$ as short-form for $(\sigma_{1} ... \, \sigma_{i})$). We also store the list of corresponding partial probabilities $\{\pi_{\rho}(\pmb{\sigma}_{[1,i]}^{\mu})\}_{\mu=1}^K$, where 
\begin{equation}
 \pi_{\rho}(\sigma_1 \cdots \sigma_i) = \sum_{\pmb\sigma \in \mathcal{P}_{N-i}} \frac{1}{2^N} \Trace[\rho \, \sigma_{1}\cdots\sigma_{i}\pmb{\sigma}]^{2} \, .
\end{equation}
Ideally, one would like to keep track of all possible sub-strings, thereby enabling the identification of those that meet the stabilizer condition $\Pi_{\rho}(\pmb{\sigma}^{\mu})=\pi_{\rho}(\pmb{\sigma}_{[1,N]}^{\mu})=1/2^N$ at $i=N$. 
Yet, this is feasible only until the number of stored sub-strings, which is $4^i$, remains under control.
In practice, one has to find effective ways of discarding certain sub-strings to ensure that their total number remains within a predefined maximum number $\nsamp$. Naturally, the goal is to keep in memory only those sub-strings that have a higher likelihood of resulting into stabilizer strings at the end of the sweep. To this purpose, we adopt the following two strategies:
\begin{itemize}
\item[$i)$] 
We notice that for any stabilizer string $\pmb{\sigma} \in \StabGr(\ket{\psi})$ the partial probability at site $i$ is lower bounded by $1/(2^{i} \chi_i)$, i.e.\ $\pi_{\rho}(\pmb{\sigma}_{[1,i]}) \geq 1/(2^{i} \chi_i)$ (see Lemma~\ref{lemma:partialstabprob} in Supp. Mat.). Accordingly, one can discard all stored sub-strings for which $2^{i} \chi_i \pi_{\rho}(\pmb{\sigma}_{[1,i]}) <  1$. 
\item[$ii)$] 
When $K$ exceed $\nsamp$, one can simply sort the probabilities $\pi_{\rho}(\pmb{\sigma}_{[1,i]}^{\mu})$ in descending order and select the sub-strings corresponding to the highest $\nsamp$ values~\cite{chertkov2022optimization}. Indeed, these are the sub-strings with the highest likelihood to maximize the final probability $\Pi_{\rho}(\pmb{\sigma})$ at the end of the sweep.
\end{itemize}
Previous points establish a straightforward method for conducting a sampling process with an enhanced ability to generate Pauli stabilizer strings. At a generic step $i$, one has to compute the conditional probabilities $\pi(\alpha | \mu)=\pi_{\rho}(\sigma^{\alpha}|\pmb{\sigma}_{[1,i-1]}^{\mu})$ ($\alpha \in \{0,1,2,3\}$, $\mu \in \{0,1,... \, K\}$) and, if their total number $4K$ exceed $\nsamp$, rules $i)$ and $ii)$ are applied to select $\nsamp$ optimal sub-strings. These have indices $(\alpha_{\filledstar}, \mu_{\filledstar})$. All other possible choices are discarded. \\
The $\pi(\alpha | \mu)$ are obtained through the tensor contraction: $\pi(\alpha | \mu) = \frac{1}{2} \, \sum_{s',s, r', r} (\sigma^{\alpha})_{s' s} (\sigma^{\alpha,*})_{r' r} \Trace[ (\mathbb{A}^{s'}_{i})^{\dag} \mathbb{L}^{\mu} \mathbb{A}^{s}_{i}  (\mathbb{A}^{r}_{i})^{\dag} \cdot$ $\cdot (\mathbb{L}^{\mu})^{\dag} (\mathbb{A}^{r'}_{i})]$ (see Fig.~\ref{fig:1} $a)$), using a set of environment matrices $\mathbb{L}^{\mu}$ (as in Ref.~\cite{Lami_2023}). These serve to encode information regarding samples collected from previously visited sites, and are updated after the selection of optimal sub-strings as: $\mathbb{L}^{\mu}  \rightarrow 1/\big(2 \pi_{\rho}(\sigma^{\alpha_{\filledstar}}|\pmb{\sigma}^{\mu_{\filledstar}}_{[1,i-1]} )\big)^{1/2} \cdot  \sum_{s',s} (\sigma^{\alpha_{\filledstar}})_{s' s} (\mathbb{A}^{s'}_{i})^{\dag} \mathbb{L}^{\mu_{\filledstar}} \mathbb{A}^{s}_{i}$ (see Fig.~\ref{fig:1} $b)$). The prefactor ensures a correct normalization, i.e.\ $\Trace[\mathbb{L}^{\mu}(\mathbb{L}^{\mu})^{\dag}] =1$ at each step. 
Initially, $K$ is set to $1$ and $\mathbb{L}^{\mu}=(1)$.

We summarize the full stabilizer sampling recipe in Algorithm~\ref{alg:stabMPS}, graphically supported by Fig.~\ref{fig:1}. Here, we represent $\mathbb{L}$ has a tensor with three indices: two indices for the auxiliary MPS space, plus $\mu$. The Pauli tensor $\sigma^{\alpha}_{s' s} = \braket{s'|\sigma^{\alpha}|s}$ has also three indices. Indices $\alpha, \mu$ are merged together at the end of each step. The output of Algorithm~\ref{alg:stabMPS} is a set of $K \leq \nsamp$ stabilizer strings. 

In order to find the generators of $\StabGr(\ket{\psi})$, one has to extract a minimal set of independent Pauli generators out of them. This task can be conveniently performed by applying Gaussian elimination on the tableau matrix obtained with the replacements $\sigma^0 \rightarrow (0,0)$, $\sigma^1 \rightarrow (1,0)$, $\sigma^2 \rightarrow (1,1)$, $\sigma^3 \rightarrow (0,1)$ from the list of samples~\cite{Nielsen_chuang_2010}. This tableau has shape $K \times 2N$ and its reduction can be performed at cost $\mathcal{O}(\nsamp N^2)$, if $N<\nsamp$. In practice, these operations are extremely fast since they only involves bitwise operations. The final result is a set of independent generators of $\StabGr(\ket{\psi})$.\\

\begin{algorithm}[H]
\caption{Stabilizer sampling from MPS}\label{alg:stabMPS}
\begin{flushleft}
\hspace*{\algorithmicindent} \textbf{Input}: a right-normalized MPS $\ket{\psi}$ of size $N$
\end{flushleft}
\begin{algorithmic}[1]
\State Initialize: $K=1$, $\{\mathbb{L}^{\mu}\}_{\mu=1}^K=\{ (1) \}_{\mu=1}^K$, $\{ \Pi^{\mu} \}_{\mu=1}^K  = \{ 1 \}_{\mu=1}^K$. 
\For{($i=1$, $i=N$, $i++$)}
     \State Compute $\pi(\alpha | \mu)=\pi_{\rho}(\sigma^{\alpha}|\pmb{\sigma}_{[1,i-1]}^{\mu})$
     \Statex \hspace{4 mm} for $\alpha \in \{ 0,1,2,3 \}$ and $\mu \in \{ 1,2 ... \, K \}$
     \State Select the $(\alpha_{\filledstar}, \mu_{\filledstar})$ s.t. $\pi(\alpha_{\filledstar} | \mu_{\filledstar}) \Pi^{\mu_{\filledstar}} \geq 1/(2^i \chi_i)$
     \State Set $K = \min \big(|\{ (\alpha_{\filledstar}, \mu_{\filledstar}) \}|, \nsamp \big)$
     \State Select $K$ indices $(\alpha_{\filledstar}, \mu_{\filledstar})$ corresponding to \Statex \hspace{4.3 mm} largest values of 
              $\pi(\alpha | \mu) \Pi^{\mu}$, discard the others
     \For{($\mu=1$, $\mu=K$, $\mu++$)}
           \State Set $\{ \pmb{\sigma}_{[1,i]}^{\mu} \} = \{ (\pmb{\sigma}_{[1,i-1]}^{\mu_{\filledstar}}, \sigma^{\alpha_{\filledstar}}) \}$
           \State Update $\{ \Pi^{\mu} \}\rightarrow,  \{ \pi(\alpha_{\filledstar} | \mu_{\filledstar}) \Pi^{\mu_{\filledstar}} \}$ and $\{ \mathbb{L}^{\mu} \}$
     \EndFor
\EndFor
\end{algorithmic}
\begin{flushleft}
\hspace*{\algorithmicindent} \textbf{Output}: $K \leq \nsamp$ stabilizer Pauli string $\{\pmb{\sigma}^{\mu} \}_{\mu=1}^K$
\end{flushleft}
\end{algorithm}

\begin{figure}[t!]
\includegraphics[width=0.8\linewidth]{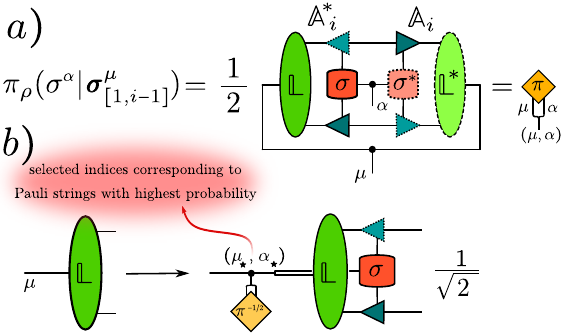}
\caption{The stabilizer sampling Algorithm \ref{alg:stabMPS}. At a generic step $i$, the probabilities $\pi(\alpha | \mu)=\pi_{\rho}(\sigma^{\alpha}|\pmb{\sigma}_{[1,i-1]}^{\mu})$ are computed by means of the stored environment matrices $\mathbb{L}$ ($a$). Afterwards, indices $(\alpha_{\filledstar}, \mu_{\filledstar})$ corresponding to the highest partial probabilities are selected and the tensor $\mathbb{L}$ is updated ($b$) \label{fig:1} }
\end{figure}

\paragraph{Iterations over modified states. --}
While these strategies are already effective in ensuring a favorable probability of sampling $\StabGr(\ket{\psi})$, this can be further increased.
Firstly, after completing the sampling sweep from $i=1$ to $i=N$, one can repeat it in the reverse direction, i.e.\ from $i=N$ to $i=1$. Before of the reversed sweep one has to put the MPS in the left-normalized gauge~\cite{Schollwock_2011}.
Secondly, one can repeat the sampling with modified states $\ket{\psi'} = U_{\mathcal{C}} \ket{\psi}$, where $U_{\mathcal{C}} \in \mathcal{C}_{N}$. Indeed, the dimension of the stabilizer group is not altered by Clifford unitaries, whereas all conditional probabilities are reshuffled. This enables the Algorithm to effectively target unsampled regions of $\StabGr(\ket{\psi})$. In practice, one can set $U_{\mathcal{C}}$ as a random Clifford circuit of depth $D$ and evolve the MPS to obtain $\ket{\psi'}$. $D$ should remain small to avoid excessively increasing the bond dimension of $\ket{\psi'}$ (which grows as $\exp(D)$). Notice that once the sampling of $\ket{\psi'}$ is completed one has to map back the sampled Pauli strings in order to find stabilizer strings of the original state, since $\StabGr(\ket{\psi}) = U_{\mathcal{C}} \StabGr(\ket{\psi'}) U_{\mathcal{C}}^{\dag}$. This task can be easily achieved at cost $\mathcal{O}(N^2)$ in the tableau formalism, because it only involves applying Clifford unitaries~\cite{Aaronson_2004, Stim_2021}. In practice, one can iterate the sampling of modified states $\ket{\psi'}$ for several random $U_{\mathcal{C}}$. At each iteration, one has to collect the newly sampled stabilizer Pauli strings, incorporate them into the tableau containing previously sampled generators and apply Gaussian elimination to find a new set of independent generators of $\StabGr(\ket{\psi})$. The total number of these, which we dub $k$, can only increase with each iteration (converging to the true value $\StabDim_{\psi}$).\\
%\textcolor{red}{In the following numerical experiments $D=1$ or $D=2$.} \\

\begin{figure}[t!]
\includegraphics[width=1.\linewidth]{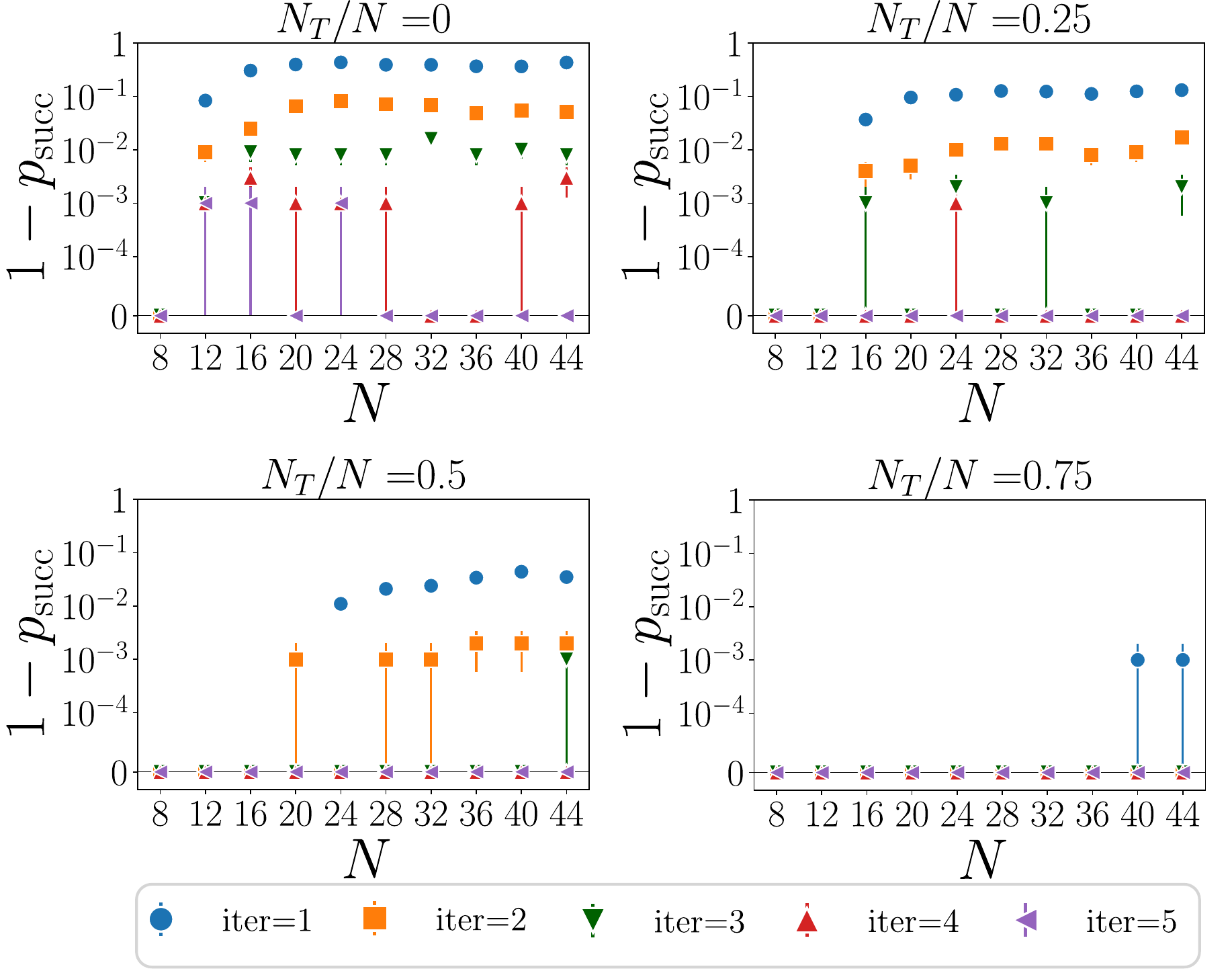}
\caption{(One minus) the probability of correctly collecting all $N-N_T$ stabilizer generators of $\ket{\psi} = U_{\mathcal{C}} \ket{N,N_T}$. Different symbols refer to iterations $1,2,3,4,5$ over modified states. 
\label{fig:2} }
\end{figure}

\begin{figure}[t!]
\includegraphics[width=1.\linewidth]{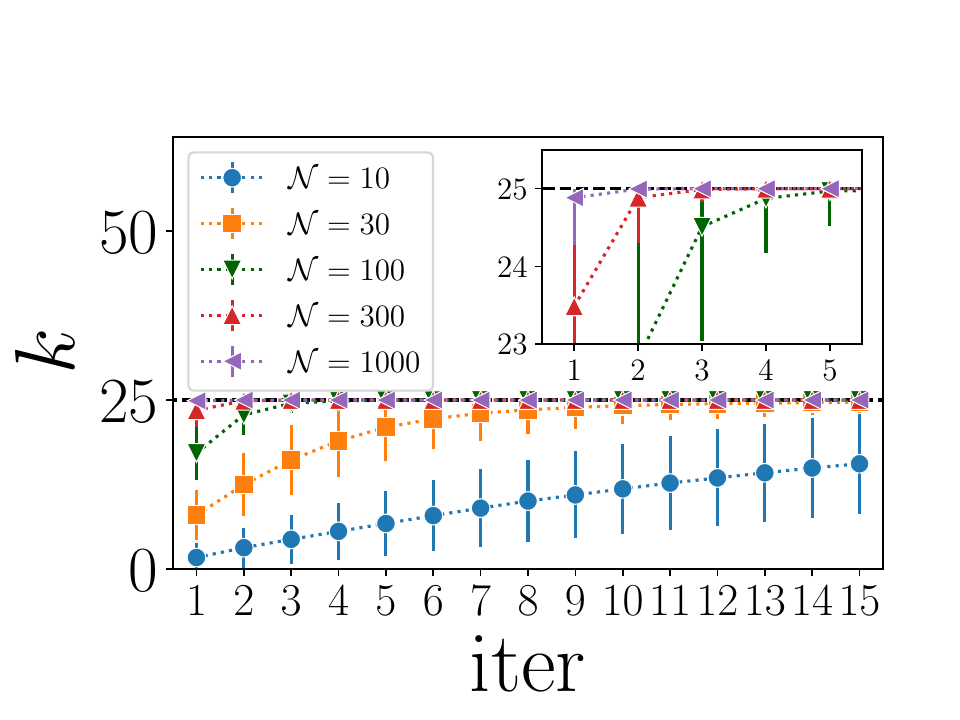}
\caption{Number of discovered generators $k$ for successive iterations iter of the alghorithm over modified states and different sample sizes $\nsamp$ (here $N=50$, $N_T=25$, $\StabDim_{\psi}=25$).
\label{fig:3} }
\end{figure}

\paragraph{Numerical experiments. --} 
%To benchmark our method, 
We prepare the state $\ket{N,N_T} \equiv \ket{0}^{\otimes N-N_T}\ket{T}^{\otimes N_T}$, where $\ket{T} = \Tgate ( \Hgate \ket{0}) = (\ket{0} + e^{i \pi/4} \ket{1})/\sqrt{2}$ is a single-qubit magical state. Afterwards, we apply a random Clifford circuit $U_{\mathcal{C}} \in \mathcal{C}_{N}$ obtaining $\ket{\psi} = U_{\mathcal{C}} \ket{N,N_T}$. By construction $\StabDim_{\psi} = N -N_T$, and the stabilizer generators of $\ket{\psi}$ can be obtained by evolving the generators $\{ \sigma^{3}_1 ... \, \sigma^{3}_{N-N_T} \}$ of $\StabGr(\ket{N,N_T})$ within the tableau formalism. $U_{\mathcal{C}}$ has depth $N$, and each layer consists of gates selected randomly from the Clifford generators $\{ \Hgate, \Sgate, \CNOT \}$. We contract the circuit to obtain $\ket{\psi}$ as an MPS, and we apply our method to detect its stabilizer generators. 

In Fig.~\ref{fig:2}, we show (one minus) the probability $p_{\text{succ}}$ of correctly obtaining $k = \StabDim_{\psi}$ as $N$ increases, for $\nsamp=10^3$ and different values of $N_T/N$. $p_{\text{succ}}$ is assessed through the iteration of the method over $10^3$ realizations of $U_{\mathcal{C}}$. For each case, we consider $5$ iterations over modified states (with $D=1$). Notice that at the final iteration, for all values of $N_T/N$, we achieve $p_{\text{succ}} \simeq 1$ within the statistical uncertainty, meaning that our technique is always able to learn entirely $\StabGr(\ket{\psi})$. In Fig.~\ref{fig:3}, we represent the number $k$ of generators found by our method as a function of subsequent iterations (iter=$1,2,3,...$) over modified states for various sample sizes $\nsamp$. We set $N=50$ and $N_T=25$, so that $\StabDim_{\psi}=N-N_T=25$, and we examine $10^3$ realizations of $U_{\mathcal{C}}$.
In this case, the bond dimension increase up to $\chi \sim 64$. We observe that even with $\nsamp\sim o(10)$, performing around 10 iterations over modified states is enough to learn the complete stabilizer group.

Afterwards, we examine a doped circuit consisting of random Clifford layers, interleaved with layers containing a constant number $\tau$ of $\Tgate$ gates placed on random sites. Clifford layers have a staircase geometry, and gates are uniformly sampled from the two-qubit Clifford group~\cite{Stim_2021}. We consider the initial state $\ket{0}^{\otimes N}$ and we investigate how its stabilizer group, which has dimension $N$, is reduced by the application of $\Tgate$ gates. In Fig.~\ref{fig:4}, we show the number of generators $k$ as a function of the discrete circuit time $n=0,1,2,3...$ for $\tau=3$ and $N=15,30,45$. Data are averaged over many circuit realizations (trajectories). 
The bond dimension $\chi$ increase up to $\chi \sim 1024$ for the larger system size. Dashed lines represent the minimum possible number of generators at time $n$, namely $k_{\text{min}}(n) = N - n \tau$ (see Lemma~\ref{lemma:stabdim} in Supp. Mat.). Data show that in typical circuit realizations the value of $k$ at step $n$ fluctuates above this line. This leads to a prolonged preservation of certain stabilizer symmetries for a time $n$ longer than the theoretically required minimum time $n_{\text{min}} = N/\tau$. A suitable rescaling of $n$ and $k$ (average) with $N$ and $n_{\text{min}}$ respectively (see inset) suggests that this effect might vanish in the thermodynamic limit, whereas fluctuations of $k$ might still be relevant.\\

\begin{figure}[t!]
\includegraphics[width=1.\linewidth]{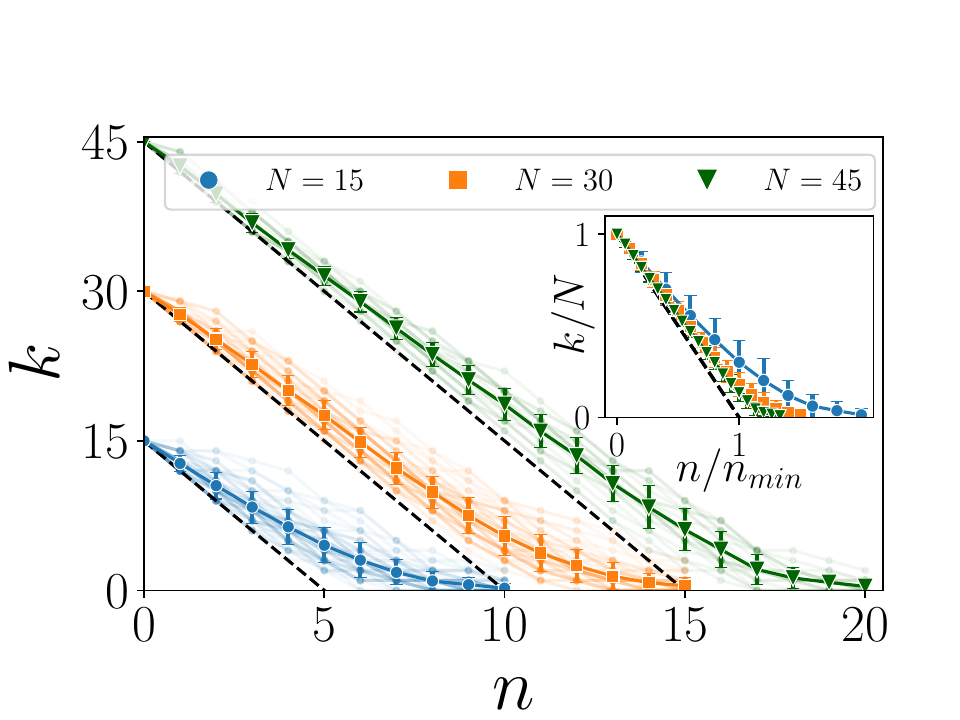}
\caption{Number of stabilizer generators $k$ for a state evolving through a doped quantum circuit, which comprises a random Clifford layer followed by a layer of $\tau=3$ $\Tgate$ gates. Integer $n$ is the discrete time of the circuit (i.e.\ $n=1$ after applying a random Clifford layer followed by a $\Tgate$ layer). Pales represent single trajectories ($N_{\text{traj}}=60$ for $N=15,30$, $N_{\text{traj}}=20$ for $N=45$), bold lines averages. Inset: same plot with $k$ rescaled by $N$ and $n$ rescaled by $n_{\text{min}}$. 
\label{fig:4} }
\end{figure}

\paragraph{Conclusions and outlook. --}
We introduced an effective classical method to learn the stabilizer group of a given Matrix Product State. In our approach, the stabilizer strings are extracted via a biased sampling in the Pauli (Bell) basis over the MPS. During the sampling, on the flight, we discard all Pauli sub-strings $\pmb{\sigma}_{[1,i]}$ at site $i$ such that \ $\pi_{\rho}(\pmb{\sigma}_{[1,i]}) < 1/(2^{i} \chi_i)$, relying on an exact theoretical argument for which we provide a proof. Manipulating our tensor networks has a computational complexity of $\mathcal{O}(\chi^3)$, implying that there are no severe constraint on the MPS bond dimension. In comparison, our method surpasses approaches based on the MPS exact replica trick, where the scaling would be at least $\mathcal{O}(\chi^6)$ (due to the bond dimension of the two replica MPS being $\chi^2$). We have shown the effectiveness of our algorithm in $T$-doped states after information-scrambling induced by a Clifford circuit. 
In addition, we analysed a prototypical case of chaotic non-equilibrium dynamics induced by the interplay of local magic gates and entangling Clifford layers. For different system sizes and up to bond dimension $\chi\sim 10^3$, we studied the dynamical depletion of the stabilizer group. We have shown that the signature of {\it stabilizerness}, which is theoretically lower-bonded by $N-\tau n$ during the discrete time-steps $n$, definitively survives for longer times with non-trivial fluctuations over the trajectories.

Finally, as further investigations, let us mention that the stabilizer group generator of an MPS can have groundbreaking consequences in decreasing the computational complexity of the state representation itself. This may lead to the development of hybrid MPS-stabilizer techniques, introducing innovative methods for classically simulating highly-nontrivial quantum states. \\

\paragraph{Acknowledgments. --}
We are particularly grateful to L. Piroli, A. Paviglianiti, G. Fux, M. Dalmonte, E.Tirrito and P. Tarabunga for inspiring discussions and for collaborations on topics connected with this work. This work was supported by the PNRR MUR project PE0000023-NQSTI, and by the PRIN 2022 (2022R35ZBF) - PE2 - ``ManyQLowD''. G.L. were partially founded by ANR-22-CPJ1-0021-01.\\

\paragraph{Note added. --}
While completing this manuscript, we became aware of a parallel independent work on non-stabilizerness and tensor networks by Tarabunga, Tirrito, Ba\~nuls and Dalmonte, where they make use of a different approach, based on compressed-MPS folding technique, to reduce the cost of the exact replica trick in computing nonstabiliserness. The work will appear on the same arvix post. \\

\bibliography{bib}

\newpage

\onecolumngrid
\appendix
\appendixtitleon
\appendixtitletocon

\begin{appendices}
\section{SUPPLEMENTARY MATERIALS}

Here, we provide proofs of mathematical facts used in the main text. 

\begin{lemma}\label{lemma:stabdim}
The stabilizer dimension $\StabDim_{\psi}$ of a $t-$doped state $\ket{\psi}$ is at least $N-t$.
\end{lemma}
\begin{proof} We revisit the proof given in Ref.~\cite{Grewal_2023}. The idea is to use induction on $t$. If $t=0$, the statement is self-evident, since $\ket{\psi}$ is a stabilizer. Now, we consider the generic case in which $\ket{\psi} = U_{\mathcal{C}} \Tgate \ket{\phi}$, where $U_{\mathcal{C}}$ is a Clifford unitary and $\ket{\phi}$ is a $(t-1)$-doped states. Thanks to the inductive hypothesis, we assume $\StabDim_{\phi} \geq N - (t-1)$. Since $U_{\mathcal{C}}$ does not change the stabilizer dimension, we have to show that the stabilizer rank of the state $\Tgate \ket{\phi}$ is at least $N-t$. To this purpose, let us consider the stabilizer group $\StabGr(\ket{\phi}) = \braket{\pmb{\sigma}_1 ... \, \pmb{\sigma}_{\StabDim_{\phi}} }$ and observe that every Pauli string $\pmb{\sigma} \in \StabGr(\ket{\phi})$ that commutes with $\Tgate$ is also a stabilizer string for $\Tgate \ket{\phi}$, since
\begin{equation}
    \pmb{\sigma} \Tgate \ket{\phi} = \Tgate \pmb{\sigma} \ket{\phi} = \pm \Tgate \ket{\phi} \, .
\end{equation}
How many elements of $\StabGr(\ket{\phi})$ commute with $\Tgate$? Certainly, all the $\pmb{\sigma} \in \StabGr(\ket{\phi})$ containing either $\sigma^{0}$ or $\sigma^{3}$ on the site $i \in \{1,2, ... \, N\}$ in which $\Tgate$ is applied. Now, a generic Pauli string $\pmb{\sigma}$ of $\StabGr(\ket{\phi})$ can be decomposed as
\begin{equation}
\pmb{\sigma} = \prod_{\mu=1}^{\StabDim_{\phi}}(\pmb{\sigma}_{\mu})^{\alpha_{\mu}} \, ,
\end{equation}
where the index $\mu=1,2 ... \, \StabDim_{\phi}$ labels here the generators of $\StabGr(\ket{\phi})$ and the exponents $\alpha_{\mu}$ are either $0$ or $1$. We can order the generators $\pmb{\sigma}_{\mu}$ in such a way that the first $k_0$ generators have $\sigma^0$ on site $i$, the subsequent $k_1$ generators have $\sigma^1$ on site $i$, the subsequent $k_2$ generators have $\sigma^2$ on site $i$ and the last $k_3$ generators have $\sigma^3$ on site $i$ ($\StabDim_{\phi} = k_0 + k_1 + k_2 + k_3$). The component of $\pmb{\sigma}$ on $i$ will be
\begin{equation}\label{eq:componentisigma}
\sigma_i = (\sigma^0)^{\nu_0}(\sigma^1)^{\nu_1}
(\sigma^2)^{\nu_2}(\sigma^3)^{\nu_3}
\end{equation}
where the exponents $\nu_0, \nu_1, \nu_2, \nu_3$ are obtained by taking the sum $\text{mod } 2$ of the $\alpha_{\mu}$ over all generators containing respectively the Pauli matrices $\sigma^{0}, \sigma^{1},\sigma^{2}, \sigma^{3}$ on $i$, i.e.\ 
\begin{equation*}
\nu_0 = \sum_{\mu=1}^{k_0} \alpha_{\mu} \Mod{2} \, , \quad \nu_1 = \sum_{\mu=k_0 + 1}^{k_0 + k_1} \alpha_{\mu} \Mod{2} \, , \quad \nu_2 = \sum_{\mu=k_0 + k_1 + 1}^{k_0 + k_1 + k_2} \alpha_{\mu} \Mod{2} \, , \quad \nu_3 = \sum_{\mu=k_0 + k_1 + k_2 + 1}^{\StabDim_{\phi}} \alpha_{\mu} \Mod{2} \, .
\end{equation*}
The value taken by $\nu_0$ and $\nu_3$ in Eq.\ref{eq:componentisigma} is not relevant, whereas there exist two possibilities to set $(\nu_1,\nu_2)$ in order to achieve $\sigma_i \propto \sigma^{0}$ or $\sigma_i \propto \sigma^{3}$,  namely: $(\nu_1=0,\nu_2=0)$,$(\nu_1=1,\nu_2=1)$. Let us consider the case $k_1, k_2 \neq 0$ and observe that one has 
\begin{equation}\label{eq:conteggio}
2^{k_0} \cdot 2^{k_3} \cdot 2^{k_1 - 1} \cdot 2^{k_2 - 1}
\end{equation}
ways of selecting the exponents $\alpha_{\mu}$ to obtain 
$(\nu_1=0,\nu_2=0)$, and the same number for $(\nu_1=1,\nu_2=1)$. The factors $2^{k_0} 2^{k_3}$ in Eq.\ref{eq:conteggio} are due to the fact that we are not interested in the values of the first $k_0$ and the last $k_3$ exponents $\alpha_{\mu}$, whereas we get $2^{k_1 - 1} 2^{k_2 - 1}$ due to imposing constraints on the parity of the sums $\sum_{\mu=k_0 + 1}^{k_0 + k_1} \alpha_{\mu}$ and $\sum_{\mu=k_0 + k_1 + 1}^{k_0 + k_1 + k_2} \alpha_{\mu}$. Ultimately, one has in total $2 \cdot 2^{k_0} 2^{k_3} 2^{k_1 - 1} 2^{k_2 - 1} = 2^{\StabDim_{\phi}-1}$ ways to select the exponents $\alpha_{\mu}$ in order to get $\sigma_i \propto \sigma^{0}$ or $\sigma_i \propto \sigma^{3}$. The cases in which $k_1$ or $k_2$ can be treated similarly, and the counting is always $\geq 2^{\StabDim_{\phi}-1}$. In conclusion, we found that there are at least $2^{\StabDim_{\phi}-1}$ elements of $\StabGr(\ket{\phi})$ commuting with $\Tgate$. Therefore $\StabDim_{\psi} \geq N - (t-1) -1 = N-t$.
\end{proof}

\begin{lemma}\label{lemma:partialstabprob}
Given a Pauli string $\pmb{\sigma}=(\sigma_1 \cdots \sigma_i, \sigma_{i+1} \cdots \sigma_N)$ belonging to the stabilizer group $\StabGr(\ket{\psi})$ of a right-normalized MPS $\ket{\psi}$, for any fixed site $i \in \{1,2,... \, N \}$ one has 
\begin{equation}
    \pi_{\rho}(\pmb{\sigma}_{[1,i]}) \geq \frac{1}{2^{i} \chi_i} \, ,
\end{equation}
where $\rho = \ket{\psi} \bra{\psi}$, $\chi_i$ is the MPS bond dimension between sites $i$ and $i+1$ and $\pi_{\rho}(\pmb{\sigma}_{[1,i]})$ is the probability of sampling the sub-string $\pmb{\sigma}_{[1,i]} = (\sigma_{i} \cdots \sigma_{i+1})$ no matter the outcomes on the other sites.  
\end{lemma}

\begin{proof}
By hypothesis we have
\begin{equation}\label{eq:proof_1}
\Pi_{\rho}(\pmb{\sigma})=\Pi_{\rho}(\sigma_1 ... \, \sigma_i , \sigma_{i+1} ... \, \sigma_N) = \frac{1}{2^N} \braket{\psi|\sigma_1 ... \, \sigma_i \sigma_{i+1} ... \, \sigma_N|\psi}^2 \equiv \frac{1}{2^N}
\end{equation}
since $\pmb{\sigma} \ket{\psi} = \pm \ket{\psi}$. Eq.\ref{eq:proof_1} can be graphically represented as follows 

\vspace{1 cm}
\begin{figure*}[h!]
\includegraphics[width=0.6\linewidth]{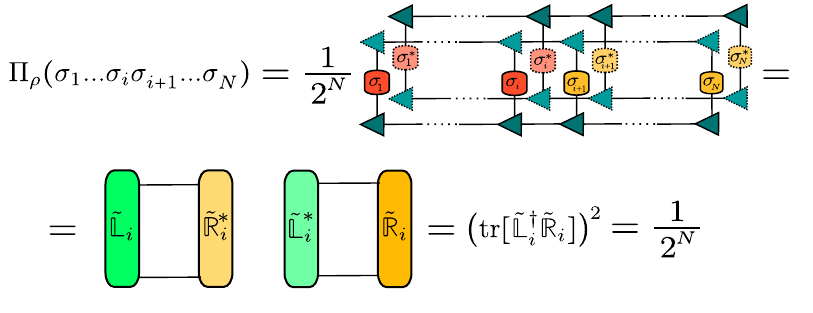}\label{fig:proof_1} 
\end{figure*}

where we introduced two environment tensors $\tilde{\mathbb{L}_i}$, $\tilde{\mathbb{R}_i}$, which are defined as \\
\begin{figure}[h!]
\includegraphics[width=0.6\linewidth]{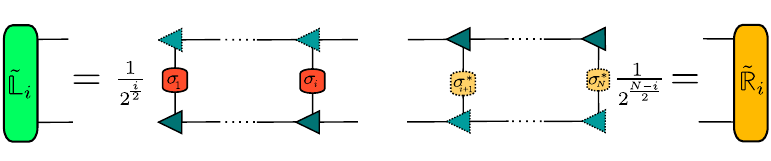}
\end{figure}

Notice that $\tilde{\mathbb{L}}_i$, $\tilde{\mathbb{R}}_i$ are defined similarly to $\mathbb{L}$, $\mathbb{R}$ at site $i$ in the main text, a part for a normalization factor. 
The probability of obtaining the sub-string $\pmb{\sigma}_{[1,i]}=(\sigma_1 \cdots \sigma_i)$ can now be written as 
\begin{figure}[h!]
\includegraphics[width=0.67\linewidth]{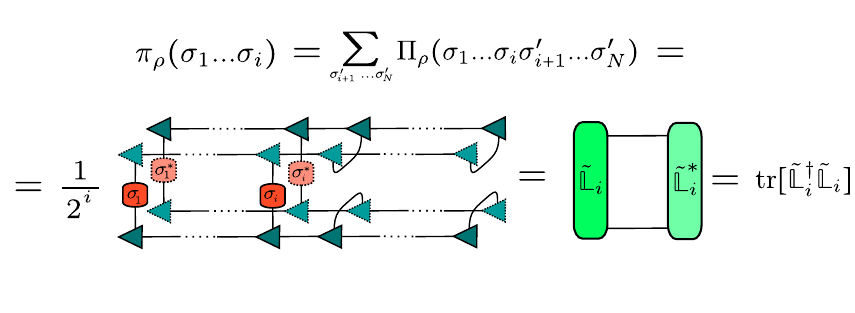}
\end{figure}

Now, by the Cauchy–Schwarz inequality, we have 
\begin{equation}
\big( \Trace[\tilde{\mathbb{L}}_i^{\dag} \tilde{\mathbb{R}}_i] \big)^2 \leq \Trace[\tilde{\mathbb{L}}_i^{\dag} \tilde{\mathbb{L}}_i] \cdot \Trace[\tilde{\mathbb{R}}_i^{\dag} \tilde{\mathbb{R}}_i] \, ,
\end{equation}
which implies 
\begin{equation}
\frac{1}{2^N} \leq \pi_{\rho}(\pmb{\sigma}_{[1,i]}) \Trace[\tilde{\mathbb{R}}_i^{\dag} \tilde{\mathbb{R}}_i] \, .
\end{equation}
We can now use the fact that $\Trace[\tilde{\mathbb{R}}_i^{\dag} \tilde{\mathbb{R}}_i] \leq \frac{\chi_i}{2^{N-i}}$, that is proven in the next Lemma~\ref{lemma:Rnorm}. We have 
\begin{equation}
\pi_{\rho}(\pmb{\sigma}_{[1,i]}) \geq \frac{1}{2^N} \frac{1}{\Trace[\tilde{\mathbb{R}}_i^{\dag} \tilde{\mathbb{R}}_i]} \geq \frac{1}{2^N} \frac{2^{N-i}}{\chi_i} = \frac{1}{2^i \chi_i} \, .
\end{equation}
\end{proof}

\begin{lemma}\label{lemma:Rnorm}
The Frobenius norm of the environment tensor $\tilde{\mathbb{R}}_i$ is upper bounded by 
\end{lemma}
\begin{equation}
    ||\tilde{\mathbb{R}}_i||_2^2 = \Trace[\tilde{\mathbb{R}}_i^{\dag} \tilde{\mathbb{R}}_i] \leq \frac{\chi_i}{2^{N-i}} \, .
\end{equation}

\begin{proof}
Here, we use the following conventions for the norm of a matrix $O$ of shape $n \times n$ having singular values $\lambda_1 \geq \lambda_2 \geq ... \, \geq \lambda_n \geq 0$~\cite{enwiki:1180861007,Watrous_2018}. The Frobenius norm is: $||O||_2^2= \Trace[O^{\dag} O] = \sum_{k=1}^n \lambda_k^2$. The operator norm is: $||O||_{\infty}= \lambda_1$. Obviously: $||O||_2^2 \leq n ||O||_{\infty}^2$. Moreover, it is well known that for any matrices $A,B$ one has 
$||A B||_{\infty} \leq ||A||_{\infty} ||B||_{\infty}$ and $||A \otimes B||_{\infty} = ||A||_{\infty} ||B||_{\infty}$. Let us now consider the MPS transfer matrix~\cite{Schollwock_2011}
\begin{equation}
 E_k = \sum_{s',s}  (\sigma^{*}_k)_{s's} 
\mathbb{A}^{s'}_{k} \otimes \big(\mathbb{A}^{s}_{k}\big)^*
\end{equation}
which can be represented as follows: \\
\begin{figure}[h!]
\includegraphics[width=0.35\linewidth]{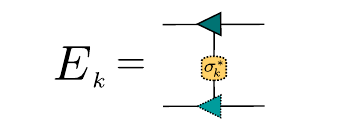}
\end{figure}
\vspace{5 mm}

The environment tensors $\mathbb{R}$ can be obtained iteratively by applying the transfer matrices $E_k$, since $\tilde{\mathbb{R}}_{k-1} = \frac{1}{\sqrt{2}} E_k[\tilde{\mathbb{R}}_{k}]$, i.e.\ 
\begin{figure}[h!]
\includegraphics[width=0.35\linewidth]{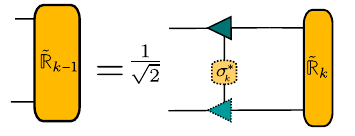}
\end{figure}

First of all, let us prove that 
\begin{equation}\label{eq:utile}
||E_k[\tilde{\mathbb{R}}_{k}]||_{\infty} \leq ||\tilde{\mathbb{R}}_{k}||_{\infty} \, .  
\end{equation} 
We perform a Singular Value Decomposition (SVD) of $\tilde{\mathbb{R}}_{k}$. We write $\tilde{\mathbb{R}}_{k} = \mathbb{U} \Lambda \mathbb{V}$, where $\mathbb{U}$ and $\mathbb{V}$ are unitary matrices (since $\tilde{\mathbb{R}}_{k}$ is a square matrix) and $\Lambda$ is a diagonal matrix containing the (non-negative) singular eigenvalues of $\tilde{\mathbb{R}}_{k}$. The action of $E_k$ over the environment tensor is therefore
\begin{equation}
 E_k[\tilde{\mathbb{R}}_{k}] = \sum_{s',s}  (\sigma^{*}_k)_{s's} 
\mathbb{A}^{s'}_{k} \, \mathbb{U} \Lambda \mathbb{V} \, \big(\mathbb{A}^{s}_{k}\big)^{\dag} \, ,
\end{equation}
which can be also recast as 
\begin{equation}
 E_k[\tilde{\mathbb{R}}_{k}] = P \big( \sigma^{*}_k \otimes \Lambda \big) Q \, ,
\end{equation}
where 
\begin{equation}
P = 
\begin{pmatrix}
\mathbb{A}^{0}_{k} \mathbb{U} & \mathbb{A}^{1}_{k} \mathbb{U} \\
\end{pmatrix}
\qquad \sigma^{*}_k \otimes \Lambda = 
\begin{pmatrix}
(\sigma^{*}_k)_{00} \Lambda & (\sigma^{*}_k)_{01} \Lambda \\
(\sigma^{*}_k)_{10} \Lambda & (\sigma^{*}_k)_{11} \Lambda \\
\end{pmatrix}
\qquad Q=
 \begin{pmatrix}
\mathbb{V} \big(\mathbb{A}^{0}_{k}\big)^{\dag}\\ 
\mathbb{V} \big(\mathbb{A}^{1}_{k}\big)^{\dag}\\
\end{pmatrix} \, .
\end{equation}
We have: 
\begin{equation}
P P^{\dag} = 
\begin{pmatrix}
\mathbb{A}^{0}_{k} \mathbb{U} & \mathbb{A}^{1}_{k} \mathbb{U} \\
\end{pmatrix}
\begin{pmatrix}
\mathbb{U}^{\dag} (\mathbb{A}^{0}_{k})^{\dag}  \\ \mathbb{U}^{\dag} (\mathbb{A}^{1}_{k})^{\dag}  \\
\end{pmatrix}
= \mathbb{A}^{0}_{k} \mathbb{U} \mathbb{U}^{\dag} (\mathbb{A}^{0}_{k})^{\dag} + \mathbb{A}^{1}_{k} \mathbb{U} \mathbb{U}^{\dag} (\mathbb{A}^{1}_{k})^{\dag} = \sum_s \mathbb{A}^{s}_{k} (\mathbb{A}^{s}_{k})^{\dag} = \mathbb{1} \, ,
\end{equation}
where we used the right-normalization of the MPS tensors. Similarly $Q^{\dag} Q = \mathbb{1}$. Since the non zero singular eigenvalues of $P$ ($Q$) are the square roots of the eigenvalues of $P P^{\dag}$ ($Q^{\dag} Q$), we have $|| P ||_{\infty} = || Q ||_{\infty} = 1$. Thus, we can write 
\begin{equation}
     ||E_k[\tilde{\mathbb{R}}_{k}]||_{\infty} =  || P \big( \sigma^{*}_k \otimes \Lambda \big) Q ||_{\infty} \leq || P ||_{\infty} ||\sigma^{*}_k \otimes \Lambda ||_{\infty} || Q ||_{\infty} = ||\sigma^{*}_k||_{\infty} ||\Lambda||_{\infty} \, .
\end{equation}
Ultimately, we get Eq.~\ref{eq:utile} by noticing that $||\tilde{\mathbb{R}}_{k}||_{\infty} = ||\Lambda||_{\infty}$ and $||\sigma^{*}_k||_{\infty} = 1$ (for every Pauli matrix $\sigma$). 
Finally, considering $\mathbb{R}_{i}$ (which is a matrix of size $\chi_i \times \chi_i$) we obtain
\begin{equation}
     ||\mathbb{R}_{i}||_{2}^2 \leq \chi_i ||\mathbb{R}_{i}||_{\infty}^2 = \chi_i \lVert \, 
 \frac{1}{2^{(N-i)/2}} E_{i+1}[E_{i+2}[E_{i+3}[... E_N[1]]]] \, \rVert_{\infty}^2 \leq \frac{\chi_i}{2^{(N-i)}} \, ,
\end{equation}
where we repeatedly used the inequality in Eq.~\ref{eq:utile}, until we reached the boundary. 
\end{proof}

\end{appendices}

\end{document}